\documentclass[12pt, reqno]{amsart}
\usepackage{helvet}
 
\usepackage[english]{babel}
\usepackage{float,subfloat}
\usepackage{here}
\usepackage[letterpaper,top=2cm,bottom=2cm,left=3cm,right=3cm,marginparwidth=1.75cm]{geometry}
\usepackage{graphicx, subcaption}
\usepackage[colorlinks=true, allcolors=blue]{hyperref}
\usepackage{amsmath, amsthm, mathtools, amscd, amsfonts, amssymb}
\usepackage{xcolor}
\definecolor{darkgreen}{rgb}{0,0.51,0.11}
\usepackage{lmodern}
\usepackage{tikz-cd}
\newtheorem{theorem}{Theorem}[section]

\theoremstyle{definition}
\newtheorem{definition}[theorem]{Definition}
\newtheorem{example}[theorem]{Example}

\newtheorem{remark}[theorem]{Remark}
\numberwithin{equation}{section}

\newcommand{\z}{\mathcal{Z}}
\newcommand{\p}{\mathcal{P}}
\newcommand{\q}{\mathcal{Q}}
\newcommand{\Z}{\mathbb{Z}}
\newcommand{\s}{\mathcal{S}}

\title{ Extended Cellular Automata}

\author[]{P. Mehdipour}
	\address[Pouya Mehdipour]{%
		Departamento de Matem\'atica, Universidade Federal de Vi\c cosa, Brazil.
	}
	\email{pouya@ufv.br}
\author[]{M. Salarinoghabi}
	\address[Mostafa Salarinoghabi]{%
		Departamento de Matem\'atica, Universidade Federal de Vi\c cosa, Brazil.
	}
	\email{mostafa.salarinoghabi@ufv.br}
\author[]{P. Gibrim}
	\address[Paula Gibrim]{%
		Departamento de Inform\'atica, Universidade Federal de Vi\c cosa, Brazil.
	}
	\email{paula.gibrim@ufv.br}



\begin{document}
\maketitle

\begin{abstract}
In this work, the one-dimensional Cellular Automaton is extended to one that involves two sets of symbols and two global rules. As a main result, the Extended Curtis-Hedlund-Lyndon Theorem is demonstrated. Such constructions can be useful in studying complex systems involving two related phenomena and provide a way to their co-study.
\end{abstract} 
\section{Introduction}
Many natural phenomena in science are modeled using differential equations. However, some phenomena are too complex to be adequately modeled by equations, and their behavior cannot be fully understood without considering the interaction of their components. Examples of such systems include the human brain, the immune system, disease transmission, communication systems, financial markets, transportation networks, ecosystems, and more. These are known as complex systems. A key characteristic of these systems is that they consist of many independent factors that interact with one another and exhibit self-organization \cite{B},\cite{13},\cite{12}. Investigating and understanding complex systems, as well as simulating the phenomena that emerge from the interaction of their components, is of central importance across many scientific fields.

Cellular automata have proven to be a highly effective approach for addressing many scientific problems, providing an efficient way to model and simulate phenomena that are difficult to capture using traditional mathematical and computational techniques. A cellular automaton is a discrete dynamical system that evolves on a grid of cells, where each cell's state is updated according to a set of local rules. The global evolution of the system is governed by a rule that determines the overall behavior of the grid. From a computational perspective, cellular automata have the same computational power as universal Turing machines, meaning they can simulate any computation. One of their most remarkable features is the ability to model and display complex behavior arising from simple rules, making them well-suited for studying and modeling physical and natural phenomena \cite{10}, \cite{02}, \cite{2}.

Cellular automata were first introduced in the 1950s through the work of John von Neumann and Stanislaw Ulam. By the 1960s, they had become an established concept in the field of symbolic dynamics. In 1969, G. A. Hedlund published a seminal paper, which is considered a cornerstone in the mathematical study of cellular automata. The key result of this work is the Curtis-Hedlund-Lyndon Theorem, which characterizes cellular automata as continuous maps that commute with two-sided shift maps \cite{H}, \cite{2}.

In this paper, we extend the classical cellular automaton framework using symbolic dynamics to introduce a new generation of one-dimensional cellular automata that involve two global rules (see Figure \ref{fig:principal} for an example). As a main result, we present the Extended Curtis-Hedlund-Lyndon Theorem, providing a mathematical characterization of extended cellular automata in terms of their symbolic dynamics.
\begin{figure}[h]
\begin{center}
\includegraphics[width=0.25\textwidth]{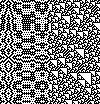}
\caption{An example of a new generation of one-dimensional CA.}
\label{fig:principal}
 \end{center}
\end{figure}
The Zip-CA construction can be useful for studying complex systems that involve two interrelated phenomena. In some future works, we will generalize this extended symbolic dynamics and develop cellular automata that incorporate any finite number of evolution rules \cite{MLS1}, \cite{MLS2}. 
\section{Zip-Shift space}
For natural numbers $k,m$ with $k\leq m$, consider two finite sets of alphabets $\z = \{a_1,\ldots,\,a_k\}$ and $\s= \{0,\ldots,\,m-1\}$.  
Let  $\Sigma_{\s}\subseteq \{(x_i)_{i\in \Z},\, x_i\in \s\}$ be a shift space (for details on shift space see for instance \cite{3}).
For any integer $\ell \geq k$, a finite sequence 
\[\omega_{[k,\ell]}=(x_k\, x_{k+1}\, \ldots\, x_{\ell}),\]
 which occurs in some point $(x_i)_{i\in \Z}\in \Sigma_{\s}$
is called a \textit{word} or \textit{block} of length $\left|\omega_{[k,\ell]}\right|=\ell-k+1$.
Now, for every $n\geq 1$, assume that $B_n^{\s}:=B_n^{\s}(\Sigma_{\s})$ be  
the set of all admissible words of length $n$ and also 
\[B^{\s}:=B^{\s}(\Sigma_{\s})=\bigcup_{n\geq 1}B_n^{\s}.\]
Note that obviously, $B_1^{\s}=\s$.
\begin{definition}
Using notations stated above, let  $\tau: \s\longrightarrow \z$ be a surjective transition map, which is not necessarily invertible. 
Also, let $Y=\{y=(y_i)_{i\in \Z},\, y_i\in \s\}$ be a two-sided shift space with its associated shift map $\sigma$. For any point $y\in Y$, we correspond a 
point $x=(x_i)_{i\in \Z}=(\ldots,x_{-2},x_{-1};\,x_0,\,x_1,\,x_2,\ldots)$ such that
\begin{equation}\label{def:zip}
x_i=\left\{\begin{tabular}{ll}
$y_i\in \s$ & $i\geq 0$\\
$\tau (y_i)\in \z$& $i<0$. 
\end{tabular}\right.
\end{equation}
Then consider the following space which:
\[
\Sigma_{\z,\s}:= \{x=(x_i)_{i\in \Z}\ : x_i \textrm{\ satisfies} \,(\ref{def:zip})\}.
\]
The space $\Sigma_{\z,\s}$ is so called the ``zip-shift space". 

If $Y$ is a full shift space (i.e., contains all possible admissible words on $\s$), then $\Sigma_{\z,\s}$  will be called  \textit{full zip-shift space}.
\end{definition} 
One important property that a new space should have is invariance under a certain kind of maps, which we define as zip-shift maps. The following definition addresses this concept.

\begin{definition}\label{def:zip shift map}
Following notations stated above, we define the zip-shift map 
\begin{eqnarray*}
\sigma_{\tau} : \Sigma_{\z,\s}&\to& \Sigma_{\z,\s} \\
x=(x_i) &\mapsto& \left(\sigma_{\tau}(x)\right),
\end{eqnarray*}
where for every $i\in\Z$,
\[
\left(\sigma_{\tau}(x)\right)_i = \begin{cases}
\tau(x_{0}) &  i= -1\\
x_{i+1} & i \ne -1,
\end{cases}
\]
which is obviously  well-defined.

\end{definition}
It is well-known that the shift map is homeomorphism. In the case of zip-shift map we have the following affirmation.
\begin{theorem}\label{teo:local homo}
The zip shift map, defined in Definition \ref{def:zip shift map}, is a local homeomorphism which is a generalization of a two-sided shift homeomorphism.
\end{theorem}
To prove this theorem, we first need to establish a topology on the zip-shift space. For this purpose, we provide the following definitions and statements:

1) When $\z=\s$ and $\tau(x)=id(x) =x$, then the zip-shift map represents the known two-sided shift map.  As cellular automaton is usually defined on full shift spaces, so we assume that the zip-shift spaces are full zip-shift spaces.  The interested reader can find more details on these concepts in \cite{L}. 
\begin{remark}\label{rmk:1}
Note that by Definition \ref{def:zip}, a full zip-shift map with transition function $\tau$ looks like a shift map for all $i\neq-1$ (i.e., $[\sigma_{\tau}(x)]_i=x_{i+1}$) and  $x_{-1}=\tau(x_0)$.
\end{remark}
2) Now, we equip $\Sigma_{\z,\s}$ with the following metric:\\
First, consider the map 
\[D:\Sigma_{\z,\s}\times \Sigma_{\z,\s}\rightarrow\mathbb{N}\cup \{0\}\]
given by, 
\begin{equation*}
D(x,y)=\begin{cases}
\infty & x=y, \\
\min\{\vert i\vert;\,(x_{i})\neq (y_{i}) \} & x\neq y.
\end{cases}
\end{equation*}

For any $x,y\in\Sigma_{\z,\s}$ the positive function
\begin{equation}\label{def:met}
\bar{d}(x,y):=\frac{1}{2^{D(x,y)}},
\end{equation}
defines a metric on $\Sigma_{\z,\s}$. This metric induces a topology on $\Sigma_{\z,\s}$ which is equivalent to the product topology (known also as Cantor topology). From now on, we consider the metric space $(\Sigma_{\z,\s}, \bar{d})$ with the induced metric topology. 
\begin{definition}
    \item (I) Any closed subset, $\Sigma$, of $\Sigma_{\z,\s}$ which is invariant under $\sigma$ (i.e., $\sigma_{\tau}^{\pm}(\Sigma)=\Sigma$) is called a \textit{sub zip-shift space} of $\Sigma_{\z,\s}$. From now on, $\Sigma$ denotes sub zip-shift space.
    \item (II) The set of all words of length $n$ in $\Sigma_{\z,\s}$, denoted by $B_n^{\z,\s}=B^{\z,\s}_n(\Sigma)$ is defined similar to the shift case. 
Indeed, the set $B^{\z,\s}=\bigcup_{n\geq 1}B^{\z,\s}_n(\Sigma)$
represents all \textit{admissible} words (in some texts it called the \textit{language}) of $\Sigma$.
    \item (III) A shift of finite type, or $SFT$, is a shift space whose set of forbidden words is finite.
\end{definition}
\begin{remark}\label{rem:tau}
Once working with admissible words of an $SFT$ space, we may need to extend the transition map $\tau:\s\to \z$ to a new map $\tilde{\tau}: \s\cup \z \to \z$, where 
\[
\tilde{\tau}(x)= \begin{cases}
\tau(x) & x\in \s,\\
x & x\in \z.
\end{cases}
\]
 In this way we define
 \[
 \tilde{\tau}(x_{[i-k,i+k]}):=\tilde{\tau}(x_{i-k})\tilde{\tau}(x_{i-k+1})\dots \tilde{\tau}(x_{i+k}).
\]
\end{remark}
An important concept that we use frequently in study of zip-shift spaces is cylinder set.  The definition of these sets is similar to the shift space. More precisely,
for a zip-shift space $(\Sigma, \overline{d})$ the basic cylinder sets are 
\begin{equation*}
C_i^{s_i}=\{x\in\Sigma\,|\,x_{i}=s_{i}, i\in\mathbb{Z}\}.
\end{equation*}
Basically, $C_i^{s_i}$ represents the set of all points in $\Sigma$, that have $s_{i}$ in their $i^{th}$ entry. Obviously,  basic cylinders are open subsets (in fact they are clopen sets) in the product topology of zip-shift spaces.  
Also, \textit{general cylinder sets} are given by 
\[
C_{i_1,\dots,i_k}^{s_{1},\dots,s_{k}}=\{ x=(x_i)\in \Sigma\,|\,x_{i_1}=s_1,\dots,x_{i_k}=s_k; i_1,\dots,i_k\in \mathbb{Z}^d\},
\]
where, for $1\leq j\leq k$, the element $s_j\in \z$ if $i_j<0$, and $s_j\in \s$ when $i_j\geq 0$.

As cylinder sets are defined independently, so we have 
\[
C_{i_1,\dots,i_k}^{s_{1},\dots,s_{k}}=C_{i_1}^{s_{i_1}}\cap C_{i_2}^{s_{i_2}}\cdots \cap C_{i_k}^{s_{i_k}}.
\] 
\begin{remark}
The set of all cylinder sets forms a basis for the topology. Moreover, $\Sigma$ is a compact subset of $\Sigma_{\z,\s}$ (Tychonoff Theorem \cite{M}).
\end{remark}
Now we are ready to prove Theorem \ref{teo:local homo}:
\begin{proof}[\textbf{Proof of Theorem \ref{teo:local homo}}]
Since the image and pre-image of the basic and general cylinder sets are cylinder sets, one concludes that $\sigma_{\tau}$ is a continuous and open map. 

Let $\bar{x}=(\cdots x_{-1};x_0x_ 1\cdots)$ be an arbitrary element of $\Sigma_{\z,\s}$. There exists a cylinder set $C=C_{-1\,0\,1}^{x_{-1}\,x_{0}\,x_1}$ such that $\bar{x}\in C$. We claim that the restriction of $\sigma_{\tau}$ to $C$ is injective.
In fact, if $\sigma_{\tau}(x)=\sigma_{\tau}(y)$ for some $\bar{y}=(\cdots y_{-2}\,x_{-1};x_0\,x_ 1\,y_2\cdots)\in C,$  then, 
\[(\cdots x_{-2}\,x_{-1}\,\tau(x_0); x_1\,x_2\cdots)=(\cdots y_{-2}x_{-1} \tau(x_0); x_1\,y_2 \cdots).\]
Therefore $x_i=y_i$ for all $i\in\mathbb{Z}.$ This means that $\bar{x}=\bar{y}$. Therefore, $\sigma_{\tau}$ is injective on $C$ and hence is a local homeomorphism. 
\end{proof}
\begin{remark}\label{rmk:2}
 Note that:
In general, the zip-shift maps defined on two sets of finite alphabets, represent finite-to-1 local homeomorphisms,  i.e., the pre-images of elements of the zip-shift space over the zip-shift map are finite. In particular, for zip-shift space $\Sigma_{\z,\s}$ with transition map $\tau: \s\to \z$, when the cardinality of the pre-images over $\sigma_{\tau}$ is equal to some fixed $n$ for all $x\in \Sigma_{\z,\s}$ we say that it is an n-to-1 zip shift map. 
\end{remark}
 In what follows we give the example of a 2-to-1 full zip shift map.
\begin{example}
Let $\s=\{0,1,2,3\}$, $\z=\{a,b\}$ and assume that the corresponded transition map $\tau:\s\to \z$ is defined as 
\[
\begin{cases}
\tau(0)=\tau(3)=a, & \\
\tau(1)=\tau(2)=b. &
\end{cases}
\]
Then for $\bar{x}=(x_i)_{i\in\mathbb{Z}}=(\cdots a\,b\,b\,b\,a\,\textbf{;}1\,0\,3\,0\,2\,\cdots)\in \Sigma$,  one can verify that 
\[\sigma_{\tau}(\bar{x})=(\cdots a\,b\,b\,b\,a\,b\,\textbf{;}\,0\,3\,0\,2\,\cdots),\] 
and 
\[\sigma_{\tau}^{-1}(\bar{x})=\{(\cdots a\,b\,b\,b\,\textbf{;}0\,1\,0\,3\,0\,2\,\cdots),(\cdots a\,b\,b\,b\,\textbf{;}3\,1\,0\,3\,0\,2\,\cdots)\}.\]
In fact for any $\bar{x}\in \Sigma_{\z,\s}$ the cardinality of $\sigma_{\tau}^{-1}(\bar{x})$, i.e. $\#(\sigma_{\tau}^{-1}(\bar{x}))=2$ 
and $\sigma_{\tau}$ is a 2-to-1 map.
\end{example} 
\section{Extended Curtis-Hedlund-Lyndon Theorem}\label{sec:Curtis theorem}
One of the fundamental results in symbolic dynamics is the Curtis–Hedlund–Lyndon Theorem. This theorem is a mathematical characterization of cellular automata in terms of their symbolic dynamics. 
In what follows first in Theorem \ref{teo:CHL_CA}, we present the classic Curtis–Hedlund–Lyndon theorem (see \cite{H}) for classic Cellular Automata and then in Theorem \ref{teo:extended curtis} we demonstrate an extended version of this theorem for generalized Cellular Automata. It is worth mentioning that,  an even more general definition of extended sliding block codes based on zip-shift spaces (not necessarily full) and a general Extended Curtis–Hedlund–Lyndon are given in \cite{L}.

\begin{definition}[Shift Based-Sliding Block Codes]\label{SBC_classic} 
   For a given pair of natural numbers $n$ and $m$, let $\s$ and $\mathcal{T}$ be two finite alphabet sets (state spaces). Assume that $X\subseteq \s^{\Z}$ is a shift space and $N=m+n+1$ (usually, $m$ is called the memory and  $n$ is called the anticipation of the dynamic). A continuous map $\phi: X\to {\mathcal{T}}^{\mathbb{Z}}$ is called a \textit{Sliding Block Code (simply denoted by} $SBC$) if there exists  a code map (i.e., a local rule) $\Phi:B_N(X)\to {\mathcal{T}}$ such that, for all $i\in\mathbb{Z}$ and every $\bar{x}=(x_i)_{i\in\mathbb{Z}}\in X$, we have:
	\begin{itemize}
		\item[($1$)] $\big[\phi(\bar{x})\big]_i=\Phi\big(x_{[i-m,i+n]}\big)$,
		\item[($2$)] $\phi$ commutes with shift maps. 
		(i.e., $\forall k\in \mathbb{Z}, \sigma^k\circ\phi = \phi\circ \sigma^k$).
	\end{itemize}
	Furthermore, we call $\phi$ a \textit{ Cellular Automaton} (or simply a CA), when $\s=\mathcal{T}$ and $X=\s^{\mathbb{Z}}$.
\end{definition}
\begin{theorem}[Curtis–Hedlund–Lyndon \cite{BS}] \label{teo:CHL_CA}
 Let $\s$ be a finite symbolic set. Then $G: \s^{\mathbb{Z}}\to \s^{\mathbb{Z}}$ is a cellular automaton if, and only if, it is continuous and commutes with the shifts.
\end{theorem}
In the following, with the aim of extending this theorem, we first present the definition of zip-cellular automata.
\begin{definition}[Zip-shift Based Sliding Block Codes (Z-SBC)]\label{SBC_zip} 
Let $(\z,\s)$  and $(\q,\p)$ be two pairs of finite alphabet sets (state spaces). Assume that $(\Sigma_{\z,\s},\sigma_{\tau})$ and $(\Sigma_{\q,\p},\sigma_{\kappa})$ represent full zip-shift spaces. Set two natural numbers $n$  and $m$. Put $N=m+n+1$. A continuous map $R: \Sigma_{\z,\s}\to \Sigma_{\q,\p}$, is a \textit{zip-shift-sliding block code} (simply denoted by Z-SBC) if there exists a code map (local rule) $\phi_{R}:B_{N}(\Sigma_{\z,\s})\to \q\cup \p$ such that for every $\bar{x}=(x_i)_{i\in\Z}\in \Sigma_{\z,\s}:$
\begin{itemize}

\item[($1$)] For all $i\in\Z$,  $[R(\bar{x})]_i=\phi_{R}(x_{[i-m,i+n]});$
\item[($2$)] For $i\geq 0$, $\tilde{\kappa}(\phi_{R}(x_{[i-m,i+n]}))=\phi_{R}(\tilde{\tau}(x_{[i-m,i+n]}))$.  
\item[($3$)]  $R$ locally ``commutes" with zip-shift maps, 
    i.e., for $i\geq 0$, both \textit{``function composition equality"}:$R\circ\sigma_{\tau}^i=\sigma_{\kappa}^i\circ R$ and \textit{``set theoretic equality"}: $R(\sigma_{\tau}^{-i}(\bar{x}))=\sigma_{\kappa}^{-i}(R(\bar{x}))$ hold. 
\end{itemize} 
Diagram \eqref{R1} demonstrates these mentioned properties.
\begin{equation}\label{R1}
\begin{tikzcd}
[row sep=tiny,column sep=tiny]& & & \\ \Sigma_{\mathcal{Z},\mathcal{S}} \arrow{rr}{\sigma_{\tau}}\arrow{dd}{R} & & \Sigma_{\mathcal{Z},\mathcal{S}}\arrow{dd}{R} \\& \circlearrowright &  & \\ \Sigma_{\q,\p}\arrow{rr}{\sigma_{\kappa}} & & \Sigma_{\q,\p}.\\
\end{tikzcd}
\end{equation}
Moreover, we call $R$ a  \textit{Zip-Cellular Automaton} (simply denoted by Zip-CA) when $\z=\q$, $\s=\p$, $\sigma_{\tau}=\sigma_{\kappa}$ and consequently $(\Sigma_{\z,\s},\sigma_{\tau})=(\Sigma_{\q,\p}, \sigma_{\kappa})$. Note that in this case, the second condition given above, simplifies to $\tau(\phi_{R}(x_{[i-m,i+n]}))=\phi_{R}(\tau(x_{[i-m,i+n]}))$ for $i\geq 0.$
\end{definition}
Just to emphasize, it is worth noting that  $\sigma_{\tau}^{-i}(\bar{x})$ (Resp. $\sigma_{\kappa}^{-i}(R(\bar{x}))$) is a set in $\Sigma_{\z,\s}$ (Resp. $\Sigma_{\q,\p}$) and $R(\sigma_{\tau}^{-i}(\bar{x}))$ is simply the image of this set under $R$.
\begin{remark}
It is noteworthy that if for given integers $m$ and $n$ and element $\bar x=(x_i)_{i\in\mathbb{Z}}\in \Sigma_{\mathcal{Z},\mathcal{S}}$ one defines,
\begin{equation}\label{def:phi}
\phi_{R}(x_{[i-m,i+n]}):=\begin{cases}
x_{i+1}\in \s \quad\quad &  i\geq 0\\
\tau(x_{0})\in \z\quad\quad &\, i=-1,\\ 
x_{i+1}\in \z\quad\quad &  i<-1. 
\end{cases}
\end{equation} 
Then $R(\bar x)=\sigma_{\tau}(\bar x)$ is a Zip-CA.
\end{remark}
In what follows the Extended Curtis–Hedlund–Lyndon theorem is demonstrated.
\begin{theorem}[Extended Curtis–Hedlund–Lyndon]\label{teo:extended curtis}
 Let $\z,\s$ be two  finite alphabet sets (state spaces). Then $R: (\Sigma_{\z,\s},\tau)\to (\Sigma_{\z,\s},\tau)$ is a zip Cellular Automaton if, and only if, it is continuous and locally ``commutes" (in the sense of Definition \ref{SBC_zip})  with zip-shift maps.
\end{theorem}
\begin{proof}
($\Rightarrow$) A Zip-CA by definition is a zip-shift sliding block code (Definition \ref{SBC_classic}) which is continuous and locally ``commutes" with zip-shift maps. 

($\Leftarrow$) Let $R: (\Sigma_{\z,\s},\tau)\to (\Sigma_{\z,\s},\tau)$  be continuous and locally ``commutes" with zip-shift map $\sigma_{\tau}$. We aim to show that there exists some $m,n\in\mathbb{N}$ and a local rule  $\phi_R: B_{N}(\Sigma_{\z,\s})\to \z\cup \s$ associated with $R$, which satisfies the properties given in Definition \ref{SBC_zip} where $N=m+n+1$. Put $\bar{x}=(x_i)\in \Sigma_{\z,\s}$.

As $(\Sigma_{\z,\s},\bar d)$ is compact (see \eqref{def:met} and Remark \ref{rmk:1}), so $R$ is uniformly continuous. Therefore, for all $\epsilon>0$ there exists  $0<\delta\leq \epsilon$ such that for $\bar{y}=(y_i)\in\Sigma_{\z,\s}$ when $\bar{d}(\bar{x},\bar{y})<\delta$, then $\bar{d}(R(\bar{x}),R(\bar{y}))<\epsilon$. Let $k\in\mathbb{N}$ such that $\frac{1}{2^k}<\delta$. In particular for every $-k\leq i\leq k$, we have $[R(\bar{x})]_i=[R(\bar{y})]_i$. Therefore, we can say that $R$ depends only on $x_{-k},\dots,x_{0},\dots,x_{k}$. Let $n=m=k$ so $N=2k+1$. Define the local rule, $\phi_R: B_{N}(\Sigma_{\z,\s})\to \z\cup \s$, as 

\begin{equation}\label{Eq:LR}
\phi_R(x_{[i-k,i+k]}) = \begin{cases}
    [R(\sigma_{\tau}^i(\bar{x}))]_0 & x_{[i-k,i+k]}\in \s^{N}, \\
    \tilde{\tau}\left([R\left(\sigma^i_{\tau}(\bar{x})\right)]_0\right) & \text{otherwise}.
\end{cases}
\end{equation}
Note that as by hypothesis, $R(\sigma_{\tau}^i(\bar{x})) = \sigma_{\tau}^i(R(\bar{x}))$ so $\phi_R$ is well-defined map. To be more precise, let $\bar{x}=(x_i)$ and $\bar{y}=(y_i)$ be two elements of $\Sigma_{\z,\s}$ for which $x_{[i-k,i+k]} = y_{[i-k,i+k]}$. 
Due to the facts that $R$ is a uniformly continuous map and the way that we choose $k$, we know that $[R(\bar{x})]_i = [R(\bar{y})]_i$ for for every $-k\leq i\leq k$ with $k>1$, in particular it is true at the position $0$. 

To complete the proof, we should demonstrate that $\phi_R$ satisfies the conditions of Definition \ref{SBC_zip}. This means that, we should prove that:
\begin{enumerate}
\item For all $i\in\Z$,  $[R(\bar{x})]_i=\phi_{R}(x_{[i-k,i+k]});$
\item For $i\geq 0$, $\tilde{\tau}(\phi_{R}(x_{[i-k,i+k]}))=\phi_{R}(\tilde{\tau}(x_{[i-k,i+k]}))$.  
\item $R$ locally ``commutes" with zip-shift maps. 
\end{enumerate} 
The first and third above conditions happen obviously. For the second condition, note that using \eqref{Eq:LR} and Remark \ref{rem:tau}, when $i\geq 0$, we have
 \[\phi_R(\tilde{\tau}(x_{[i-k,i+k]}))=\tilde{\tau}([R(\sigma_{\tau}^i(\bar x))]_0) = \tilde{\tau}(\phi_{R}(x_{[i-k,i+k]})).\] 
Therefore $R$ is a Zip-CA.

\end{proof}
\subsection{Construction of Elementary Zip-Cellular Automaton}
In this section, we provide a simple way to construct the Zip-CA and show that full zip shift spaces are examples of Zip-CA. In order to provide such construction we start with a full shift space.
Let $\s=\z$ be a set of finite alphabets and consider the bilateral shift space $\Sigma_{S}$ with $\tau=id$. One assumes that $\tilde{\tau}:\z\cup \s \to \z$ is the extension of $\tau$ as in Remark \ref{rem:tau}.

\[
\begin{tikzcd}
\Sigma_{\s} \arrow [d,"R_1"]  \arrow[bend right]{dd}{\hspace*{-0.7cm}\rotatebox{0}{R}}\arrow[r,"\sigma"] & \Sigma_{\s} \arrow{d}{\hspace*{-0.6cm}{R_1}} \arrow[bend left]{dd}{\hspace*{0.1cm}\rotatebox{0}{R}}\\
\Sigma_{\s}  \arrow [d,"R_2"] \arrow[r,"\sigma"]& \Sigma_{\s}  \arrow{d}{\hspace*{-0.6cm}{R_2}}\\
\Sigma_{\z,\s} \arrow[r,"\sigma_{\tau}"]& \Sigma_{\z,\s}
\end{tikzcd}
\]

Here $R_1$ is a classical CA with a local rule $\phi_{R_1}:B_{N}(\Sigma_{\s})\to \s $, where $N=m+n+1$ for given natural numbers $m$ and $n$. Now let $\s,\z$ be two alphabet sets and $R_2$ a modified global rule with some associated local rule $\phi_{R_2}:B_{N}(R_1(\Sigma_{\s}))\to \z\cup \s$ where 
\begin{equation}\label{LR}
[R_2(\bar{x})]_i:=\begin{cases}
x_{i}\in \s \quad\quad &  i\geq 0\\
\phi_{R_2}(x_{[i-n,i+m]})\in \z\quad\quad &  i<0. 
\end{cases}
\end{equation} 

Set $R=R_2\circ R_1$. It is not difficult to verify that it satisfies the conditions of Definition \ref{SBC_zip}, i.e., it is a Zip-CA. In what follows, we have constructed examples of Elementary Zip-CAs with $\z=\s$ but $\Sigma_{\z,\s}\neq\Sigma_{\s}$. 


\subsection{Some examples of Zip-CA}

 In computer science a classic \textit{cellular automaton} usually is defined by using the following terminology and is denoted by a quintuple of the form $C=(\ell, \mathcal{S}, c_0, n, R)$. For more details see for instance \cite{2}.  
 
     \begin{enumerate}
     \item \textit{Cell Grid}: is a strip of one-dimensional cells of finite size $\ell$.
     \item \textit{Set of States}: is a finite set of states or alphabets denoted by $\mathcal{S}$.
         \item \textit{Initial Configuration:} is  a specific association of states in a cell grid, usually denoted by $c_0$.
         \item \textit{Neighborhood radius}: is the radius of a symmetric block around a cell, usually denoted by a natural number $n$.
         \item \textit{Local rule}: is an application  $R:\mathcal{S}^{2n+1} \to \mathcal{S}$ where $n$ represents the neighborhood radius.
     \end{enumerate}
     
It is worth mentioning that a configuration $c$ of a cellular automaton is an application $c:\mathbb{Z}\to \mathcal{S}$ that specifies the state of each cell in a cell grid. The set of all possible configurations of a grid is represented by $\mathcal{C}$. If $c$ is a constant function, which would lead all cells to the same state, we call it \textit{trivial configuration}. The local rules applied in the neighborhood of a cell is a dynamical system that is called the  \textbf{global transition function} $G:\mathcal{C}\to\mathcal{C}$, where $G(c) =e$ is a new setting in $\mathcal{C}$. The group of CAs in which, $n=1$ and $|\s|=2$ (e.g., $\s = \{0,1\}$)  are called \textit{Elementary Cellular Automata (ECA)}.
In mid 1980, S. Wolfram began working on cellular automata 
and represents the following classification for elementary cellular automata \cite{11} . 

\begin{itemize}
         \item[] \textbf{Class I}: Almost all initial configurations lead to a homogeneous state, where all cells reach the same value 0 \underline{or} 1, that is, they reach a trivial configuration.
         \item[] \textbf{Class II}: Almost all initial configurations lead to a stable and periodic state in time and spatially in-homogeneous, that is, not all cells have the same value. 
        
        
         \item[] \textbf{Class III}: Almost all initial configurations lead to a disordered state, with no recognizable pattern. Class III elements are known to have chaotic behavior, implying that they will be sensitive to the initial condition. 
         \item[] \textbf{Class IV}: Almost all initial configurations have temporal evolution with complex structures and with unpredictable evolution that can propagate, create and/or annihilate other structures. In particular, class IV elements is known as ECAs with chaotic behavior whose behavior is a mixture of previous classes.
     \end{itemize}
This classification is represented in Table \ref{tab:classif} \cite{6}.

\begin{table}[h]
	\caption{\label{tab:classif}Table of Wolfram classification given in \cite{6}.}
	\centering
	\begin{tabular}{|ll|}
		\hline
		\multicolumn{2}{|l|}{\hspace{4.5cm}\textbf{Classification of Wolfram}} \\ \hline
		\multicolumn{1}{|l|}{\textbf{Classes}} & \textbf{Rules} \\ \hline
		\multicolumn{1}{|l|}{\textbf{I}} & \begin{tabular}[c]{@{}l@{}}0, 8, 32, 40, 64, 96, 128, 136, 160, 168, 192, 224, 234, 235, 238, 239, \\ 248, 249, 250, 251, 252, 253, 254, 255\end{tabular} \\ \hline
		\multicolumn{1}{|l|}{\textbf{II}} & \begin{tabular}[c]{@{}l@{}}1, 2, 3, 4, 5, 6, 7, 9, 10, 11, 12, 13, 14, 15, 16, 17, 19, 20, 21, 23, 24, \\ 25, 26, 27, 28, 29, 31, 33, 34, 35, 36, 37, 38, 39, 42, 43, 44, 46, 47, \\ 48, 49, 50, 51, 52, 53, 55, 56, 57, 58, 59, 61, 62, 63, 65, 66, 67, 68, \\ 69, 70, 71, 72, 73, 74, 76, 77, 78, 79, 80, 81, 82, 83, 84, 85, 87, 88, \\ 91, 92, 93, 94, 95, 98, 99, 100, 103, 104, 108, 109, 111, 112, 113, 114, \\ 115, 116, 117, 118, 119, 123, 125, 127, 130, 131, 132, 133, 134, 138, \\ 139, 140, 141, 142, 143, 144, 145, 148, 152, 154, 155, 156, 157, 158, \\ 159, 162, 163, 164, 166, 167, 170, 171, 172, 173, 174, 175, 176, 177, \\ 178, 179, 180, 181, 184, 185, 186, 187, 188, 189, 190, 191, 194, 196, \\ 197, 198, 199, 200, 201, 202, 203, 204, 205, 206, 207, 208, 209, 210, \\ 211, 212, 213, 214, 215, 216, 217, 218, 219, 220, 221, 222, 223, 226, \\ 227, 228, 229, 230, 231, 232, 233, 236, 237, 240, 241, 242, 243, 244, \\ 245, 246, 247\end{tabular} \\ \hline
		\multicolumn{1}{|l|}{\textbf{III}} & \begin{tabular}[c]{@{}l@{}}18, 22, 30, 45, 60, 75, 86, 89, 90, 101, 102, 105, 122, 126, 129, 135, \\ 146, 149, 150, 151, 153, 161, 165, 182, 183, 195\end{tabular} \\ \hline
		\multicolumn{1}{|l|}{\textbf{IV}} & 41, 54, 97, 106, 107, 110, 120, 121, 124, 137, 147, 169, 193, 225 \\ \hline
	\end{tabular}
\end{table}


\begin{example}
Figure \ref{fig:principal} shows a Zip-CA example which is made from ECA rules R105 of class III and R110 of class IV. More figures of Zip-CA made from different Wolfram classes are given in Figure 3. 
\begin{figure}[htp!]
\begin{center}
\includegraphics[width=0.25\textwidth]{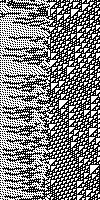}
\caption{Zip-CA: made from rules of classes III and IV.}
\label{fig:principal}
 \end{center}
\end{figure}
\begin{figure}[hbt!]
        \captionsetup[subfigure]{labelformat=empty}
	\centering
        \subfloat[\centering R248+R235]{{\includegraphics[width=2.375cm]{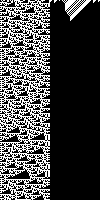}}}%
	\qquad
	\subfloat[\centering R241+R239]{{\includegraphics[width=2.375cm]{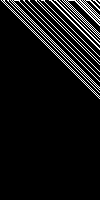} }}%
	\qquad
	\subfloat[\centering R136+R151]{{\includegraphics[width=2.375cm]{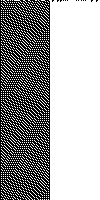} }}
	\qquad
	\subfloat[\centering R40+R147]{{\includegraphics[width=2.375cm]{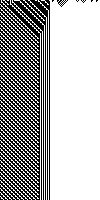}}}%
	\qquad
	\subfloat[\centering R241+R25]{{\includegraphics[width=2.375cm]{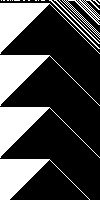} }}
	\qquad
	\subfloat[\centering  R10+R25]{{\includegraphics[width=2.375cm]{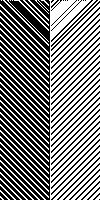} }}%
	\qquad
	\subfloat[\centering  R241+R75]{{\includegraphics[width=2.375cm]{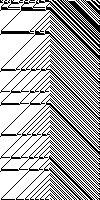} }}
	\qquad
	\subfloat[\centering  R20+R97]{{\includegraphics[width=2.375cm]{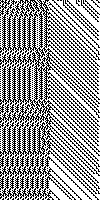} }}%
	\qquad
	\subfloat[\centering R182+R235]{{\includegraphics[width=2.375cm]{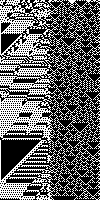}}}%
	\qquad
	\subfloat[\centering  R18+R27]{{\includegraphics[width=2.375cm]{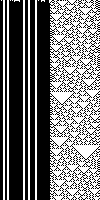} }}%
	\qquad
	\subfloat[\centering R165+R101]{{\includegraphics[width=2.375cm]{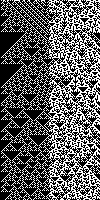} }}
	\qquad
	\subfloat[\centering  R30+R147]{{\includegraphics[width=2.375cm]{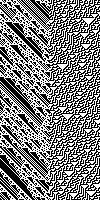}}}%
	\qquad
	\subfloat[\centering R106+R249]{{\includegraphics[width=2.375cm]{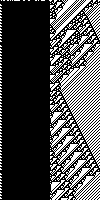} }}
	\qquad
	\subfloat[\centering  R41+R49]{{\includegraphics[width=2.375cm]{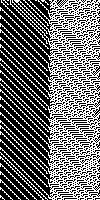} }}%
	\qquad
	\subfloat[\centering R120+R101]{{\includegraphics[width=2.375cm]{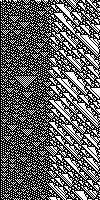} }}
	\qquad
	\subfloat[\centering R110+R225]{{\includegraphics[width=2.375cm]{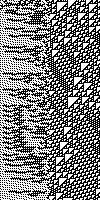} }}%
	\caption{Some Zip-CAs made from different Wolfram classes}%
\end{figure}
\end{example}

\subsection*{Data availability}
Data available on request from the authors
$\&$ 
The data that support the findings of this study are available from the corresponding author upon reasonable request.
\subsection*{Conflicts of interest}
The authors declare no conflicts of interest.
\subsection*{Acknowledgments} The authors would like to thanks Fundação de Amparo à Pesquisa do Estado de Minas Gerais (FAPEMIG) for financial supports. 
\newpage
\bibliographystyle{alpha}

\end{document}